\documentclass[11pt]{article}
\pdfoutput=1
\usepackage[utf8]{inputenc}
\usepackage[
 a4paper,
 total={160mm,257mm},
 left=25mm,
 top=20mm]{geometry}
\usepackage{hyperref}
\usepackage{amsfonts}
\usepackage{amsmath}
\usepackage{amssymb}
\usepackage{mathrsfs}  
\usepackage{color}
\usepackage{authblk}
\usepackage{physics}
\usepackage{multicol}
\usepackage{tikz}
\usepackage{setspace}
\usepackage{amsthm}
\usepackage{bbm}
\usepackage{comment}
\usepackage[toc,page]{appendix}
\theoremstyle{definition}
\newtheorem{theorem}{Theorem}[section]
\newtheorem{corollary}[theorem]{Corollary}
\newtheorem{definition}[theorem]{Definition}
\newtheorem{lemma}[theorem]{Lemma}
\newtheorem{proposition}[theorem]{Proposition}
\theoremstyle{remark}

\newtheorem{remark}[theorem]{Remark}

\newcommand{\cvev}[1]{(\, #1 \,)}
\newcommand{\rvev}[1]{\langle \, #1 \, \rangle}
\newcommand{\crvev}[1]{(\, #1 \,\rangle}
\newcommand{\rcvev}[1]{\langle\, #1 \,)}
\newcommand{\ee}{\mathrm{e}}

\newcommand{\id}{\mathbbm{1}}
\newcommand{\rem}[1]{\textbf{\textcolor{blue}{(#1)}}}

\DeclareMathOperator{\diag}{diag}

\numberwithin{equation}{section}

\title{\hspace{-.2em}\textbf{Characteristic Polynomials in Coupled Matrix Models}} 
\author{\textsc{Nicolas Babinet} and \textsc{Taro Kimura}}
\affil{%
Institut de Math\'ematiques de Bourgogne,
Universit\'e Bourgogne Franche-Comt\'e
}
\date{}

\begin{document}

\maketitle

\begin{abstract}
We study correlation functions of the characteristic polynomials in coupled matrix models based on the Schur polynomial expansion, which manifests their determinantal structure.
\end{abstract}

\tableofcontents
\vspace{1.5em}
\hrule

\section{Introduction and summary}

\subsection{Introduction}

Random Matrix Theory (RMT) has been playing an important role over the decades in the both of physics and mathematics communities~\cite{Mehta:2004RMT,Forrester:2010,Akemann:2011RMT,Eynard:2015aea}.
Applying the analogy with Quantum Field Theory (QFT), the asymptotic behavior appearing in the large size limit (large $N$ limit) is interpreted as a classical behavior as the parameter $1/N$ plays a role of the Planck constant.
From this point of view, it is an important task to explore the finite $N$ result to understand the quantum $1/N$ correction and also the non-perturbative effect beyond the perturbative analysis.
\noindent
The purpose of this paper is to show the finite $N$ exact result of a class of the correlation functions in the generalized two-matrix model, what we simply call the coupled matrix model, which contains various models coupled in the chain.
See, e.g,~\cite{Itzykson:1979fi,Eynard:1998JPA,Bertola:2001hq,Bertola:2001br,Bertola:2002xf,Bertola:2003iw,Bertola:2009CMP} and also~\cite{Eynard:2005ab,Bertola:2011,Orantin:2011} for the development in this direction.
We will show that this model can be analyzed using its determinantal structure, which is a key property to obtain the finite $N$ exact result.
In this paper, we in particular consider the correlation function of the characteristic polynomials in the coupled matrix model.
It has been known in the context of RMT that the characteristic polynomial plays a central role in the associated differential equation system through the Riemann--Hilbert problem and the notion of quantum curve.
In addition, the characteristic polynomial is essentially related to various other important observables in RMT, e.g., the resolvent, the eigenvalue density function, etc.
See, e.g.,~\cite{Morozov:1994hh,Brezin:2000CMP,Fyodorov:2002jw,Strahov:2002zu,Akemann:2002vy,Baik:2003JMP,Borodin:2006CPAM} and also~\cite{Brezin:2011} for earlier results in this direction.

\subsection{Summary of the results}

We state the summary of this paper.
In Section~\ref{sec:model}, we introduce the generalized coupled matrix model defined as the following formal eigenvalue integral,
\begin{align}
Z_N = \frac{1}{N!^2} \int \prod_{k = L, R} \dd{X}_k \ee^{- \tr V_k(X_k)} \prod_{i<j}^N (x_{k,i} - x_{k,j}) \det_{1 \le i, j \le N} \omega(x_{L,i},x_{R,j}) 
\end{align}
for arbitrary potential functions $V_k(x)$ and a two-variable function $\omega(x,y)$.
See Definition~\ref{def:Z_fn} for details.
We then show that this eigenvalue integral is reduced to the determinant of the norm matrix of the corresponding two-variable integral.
We also show that biorthogonal polynomials, which diagonalize the norm matrix, simplify the formulas.
We mention in Section~\ref{subsec:pe} that the analysis shown there is straightforwardly applied to the coupled matrix generalization of the polynomial ensemble~\cite{Kuijlaars:2014RMTA} defined for a set of arbitrary functions, containing various known models, e.g., the external source model~\cite{Brezin:2016eax}.
See also~\cite{Borodin:1998hxr}.
In Section~\ref{sec:ch_poly}, we study the correlation function for the coupled matrix model.
In Section~\ref{sec:Schur_av}, we show the Schur polynomial average, which will be a building block of the characteristic polynomials discussed throughout the paper.
In Sections~\ref{sec:ch_poly_av} and \ref{sec:ch_poly_inv_av}, we explore the correlation function of the characteristic polynomial and its inverse, and show that they are concisely expressed as a determinant of the biorthogonal polynomial and its dual.
We remark that these results are natural generalization of the earlier results on the one-matrix model case.
In Section~\ref{sec:pair_corr}, we consider the pair correlation function, which involves both the characteristic polynomials coupled with $X_L$ and $X_R$.
In this case, the correlation functions are again expressed as a determinant, while the corresponding matrix element is written using the Christoffel--Darboux (CD) kernel and its dual.

\subsection*{Acknowledgments}

We would like to thank Bertrand Eynard for useful conversation.
This work was supported in part by ``Investissements d'Avenir'' program, Project ISITE-BFC (No.~ANR-15-IDEX-0003), EIPHI Graduate School (No.~ANR-17-EURE-0002), and Bourgogne-Franche-Comté region.

\section{Coupled matrix model}\label{sec:model}

In this paper, we explore the coupled matrix model defined as follows.
\begin{definition}[Partition function]\label{def:Z_fn}
Let $V_k(x)$ $(k = L,R)$ be a polynomial function and $\omega(x_L,x_R)$ be a two-variable function.
Let $(X_k)_{k = L,R} = (x_{k,i})_{k = L, R, i = 1,\ldots, N}$ be a set of formal eigenvalues.
Then, we define the partition function of the coupled matrix model,
\begin{align}
Z_N = \frac{1}{N!^2} \int \prod_{k = L, R} \dd{X}_k \ee^{- \tr V_k(X_k)} \Delta_N(X_L) \det_{1 \le i, j \le N} \omega(x_{L,i},x_{R,j}) \Delta_N(X_R)
\, ,
\label{eq:Z_fn}
\end{align}
where we denote the Vandermonde determinant by
\begin{align}
    \Delta_N(X) = \prod_{i<j}^N (x_i - x_j) 
    \, .
\end{align}
\end{definition}
\begin{remark}
We formally consider the eigenvalues $(x_{k,i})$ as complex variables, and thus their integration contour is taken to provide a converging integral, which is not unique in general.
In this paper, we do not discuss the contour dependence on the partition function, so that we always consider the eigenvalue integral as a formal integral.
\end{remark}
\noindent
Throughout the paper, we frequently use the following identity.
\begin{lemma}[Andréief--Heine identity]\label{lemma:AH_id}
Let $(f_i(x))_{i = 1,\ldots,N}$ and $(g_i(x))_{i = 1,\ldots,}$ be the sequences of integrable functions on the domain $D$.
Denoting $\dd{X} = \dd{x}_1 \cdots \dd{x}_N$, the following identity holds,
\begin{align}
    \frac{1}{N!} \int_{D^N} \dd{X} \det_{1 \le i, j \le N} f_i(x_j) \det_{1 \le i, j \le N} g_i(x_j) = \det_{1 \le i, j \le N} \qty( \int_D \dd{x} f_i(x) g_j(x) )
    \, ,
\end{align}
which is called the Andréief--Heine (AH) identity.
\end{lemma}

\begin{proposition}[Hermitian matrix chain models]
Let $(M_k)_{k = 1,\ldots,\ell}$ be a set of $\ell$ Hermitian matrices of rank $N$.
The following matrix chain models are reduced to the coupled matrix model of the form of~\eqref{eq:Z_fn}:
\begin{subequations}
\begin{align}
    Z_{\text{pot}} & = \int \prod_{k = 1,\ldots,\ell} \dd{M_k} \ee^{-\tr V_k(M_k)} \prod_{k = 1}^{\ell-1} \ee^{\tr M_k M_{k+1}}
    \, , \\
    Z_{\text{Cauchy}} & = \int \prod_{k = 1,\ldots,\ell} \dd{M_k} \ee^{-\tr V_k(M_k)} \prod_{k = 1}^{\ell - 1} \det( M_k \otimes \id_N + \id_N \otimes M_{k+1})^{-N}
    \, .
\end{align}
\end{subequations}
We call them the potential-interacting matrix chain and the Cauchy-interacting matrix chain, respectively.
\end{proposition}
\begin{proof}
Diagonalizing each Hermitian matrix using the unitary transform for $k = 1,\ldots,\ell$,
\begin{align}
    M_k = U_k X_k U_k^{-1}
    \, ,  \qquad
    X_k = \diag(x_{k,1},\ldots,x_{k,N})
    \, , \qquad
    U_k \in \mathrm{U}(N)
    \, ,
\end{align}
the matrix measure is given by
\begin{align}
    \dd{M}_k = \frac{\dd{U_k} \dd{X}_k}{N!(2\pi)^N} \Delta_N(X_k)^2
    \, , \qquad
    \dd{X}_k = \prod_{i=1}^N \dd{x}_{k,i}
    \, ,
\end{align}
where we denote the Haar measure of each unitary matrix by $\dd{U_k}$.
We remark that the factors $(2\pi)^N$ and $N!$ are interpreted as the volumes of the maximal Cartan torus $\mathrm{U}(1)^N \subset \mathrm{U}(N)$, and the symmetric group $\mathfrak{S}_N$, which is the Weyl group of the unitary group $\mathrm{U}(N)$.

\noindent
For the potential-interacting chain, we may use the Harich-Chandra--Itzykson--Zuber formula~\cite{Itzykson:1979fi},
\begin{align}
    \int_{\mathrm{U}(N)} \dd{U} \ee^{\tr U X U^{-1} Y} = \frac{c_N}{\Delta_N(X) \Delta_N(Y)} \det_{1 \le i, j \le N} \ee^{x_i y_j}
\end{align}
where the constant factor $c_N = \Gamma_2(N+1) = \prod_{j=0}^{N-1} j!$ is chosen to be consistent with the normalization of the group integral, $\int_{\mathrm{U}(N)} \dd{U} = 1$.
Then, we obtain
\begin{align}
    Z_{\text{pot}} & = \frac{c_N^{\ell-1}}{N!^\ell} \int \prod_{k = 1,\ldots,\ell} \frac{\dd{X}_k}{(2\pi)^N} \ee^{-\tr V_k(X_k)} \Delta_N(X_1) \qty( \prod_{k = 1}^{\ell - 1} \det_{1 \le i, j \le N} \ee^{x_{k,i} x_{k+1,j}} ) \Delta_N(X_\ell)
    \nonumber \\
    & = \frac{c_N^{\ell-1}}{N!^2} \int \prod_{k = 1,\ell} \frac{\dd{X}_k}{(2\pi)^N} \ee^{-\tr V_k(X_k)}
    \Delta_N(X_1) \det_{1 \le i, j \le N} \qty( \int \prod_{k = 2,\ldots,\ell-1} \frac{\dd{x}_k}{2\pi} \ee^{-V_k(x_k)}
    \prod_{k = 1}^{\ell - 1} \ee^{x_k x_{k+1}} ) \Delta_N(X_\ell) 
    \, ,
\end{align}
where we apply the AH identity (Lemma~\ref{lemma:AH_id}) for $(X_k)_{k = 2,\ldots,\ell-1}$.
Identifying $(X_1,X_\ell) = (X_L,X_R)$ and 
\begin{align}
    \omega(x_1,x_\ell) = \int \prod_{k = 2,\ldots,\ell-1} \frac{\dd{x}_k}{2\pi} \ee^{-V_k(x_k)}
    \prod_{k = 1}^{\ell - 1} \ee^{x_k x_{k+1}}
    \, ,
\end{align}
we arrive at the expression~\eqref{eq:Z_fn} up to an overall constant.

\noindent
For the Cauchy-interacting chain, we remark the relation~\cite{Bertola:2009CMP}
\begin{align}
    \det(M_k \otimes \id_N + \id_N \otimes M_{k+1})^{-N}
    & \xrightarrow{\text{diagonalization}} \prod_{1 \le i,j \le N} \frac{1}{x_{k,i} + x_{k+1,j}} 
    \nonumber \\
    & = \frac{1}{\Delta_N(X_k) \Delta_N(X_{k+1})} \det_{1 \le i,j \le N} \qty( \frac{1}{x_{k,i} + x_{k+1,j}} )
    \, .
\end{align}
Therefore, we may write the Cauchy-interacting chain partition function as
\begin{align}
    Z_{\text{Cauchy}} & = \frac{1}{N!^\ell} \int \prod_{k = 1,\ldots,\ell} \frac{\dd{X}_k}{(2\pi)^N} \ee^{-\tr V_k(X_k)} \Delta_N(X_1) \prod_{k = 1}^{\ell-1} \det_{1 \le i,j \le N} \qty( \frac{1}{x_{k,i} + x_{k+1,j}} ) \Delta_N(X_\ell)
    \, .
\end{align}
Similarly, applying the AH identity for $(X_k)_{k = 2,\ldots,\ell-1}$, and identifying $(X_1,X_\ell) = (X_L,X_R)$ with
\begin{align}
    \omega(x_1,x_\ell) = \int \prod_{k = 2,\ldots,\ell-1} \frac{\dd{x}_k}{2\pi} \ee^{-V_k(x_k)}
    \prod_{k = 1}^{\ell - 1} \frac{1}{x_{k} + x_{k+1}}
    \, ,
\end{align}
we arrive at the expression~\eqref{eq:Z_fn}.
This completes the proof.
\end{proof}

\begin{remark}
We can in general obtain the coupled matrix model~\eqref{eq:Z_fn} from the matrix chain if the nearest-neighbor interaction is given in the determinantal form 
\begin{align}
\frac{1}{\Delta_N(X_k) \Delta_N(X_{k+1})}
\det_{1 \le i, j \le N} I(x_{k,i},x_{k+1,j})
\end{align}
after the diagonalization.
We also remark that the supermatrix model
\begin{align}
    Z_\text{susy} & = \frac{1}{N!^2} \int \dd{X} \dd{Y} \ee^{- \tr V(X) + \tr V(Y)} \Delta_N(X)^2 \Delta_N(Y)^2 \prod_{1 \le i,j \le N} (x_i - y_j)^{-2}
    \nonumber \\
    & = \frac{1}{N!^2} \int \dd{X} \dd{Y} \ee^{- \tr V(X) + \tr V(Y)} \det_{1 \le i, j \le N} \qty( \frac{1}{x_i - y_j} )^2
\end{align}
has a closed form to the partition function~\eqref{eq:Z_fn}, but it does not belong to the coupled matrix model of our current interest.
\end{remark}

\subsection{Determinantal formula}

We show that the partition function~\eqref{eq:Z_fn} is written in a determinantal form.
In order to show this, we introduce the notations.

\begin{definition}
We define the inner product with respect to the potentials $V_{L,R}(x_{L,R})$,
\begin{align}
    \cvev{ f \mid \omega \mid g} =
    \int \prod_{k = L, R} \dd{x}_k \ee^{-V_k(x_k)} f(x_L) \omega(x_L,x_R) g(x_R)
    \, .
    \label{eq:cvev}
\end{align}
For a set of arbitrary monic polynomials $(p_i(x), q_i(x))_{i \in \mathbb{Z}_{\ge 0}}$, where $p_i(x) = x^i + \cdots$ and $q_i(x) = x^i + \cdots$, we define the norm matrix,
\begin{align}
    \mathsf{N}_{i,j} = \cvev{p_{i} \mid \omega \mid q_{j}}
    \, .
\end{align}
\end{definition}

\begin{proposition}
The coupled matrix model partition function~\eqref{eq:Z_fn} is given as a rank $N$ determinant of the norm matrix,
\begin{align}
    Z_N = \det_{1 \le i, j \le N} \mathsf{N}_{N-i,N-j}
    \, .
\end{align}
\end{proposition}
\begin{proof}
Noticing that the Vandermonde determinant is written as a determinant of arbitrary monic polynomials,
\begin{align}
 \Delta_N(X_L) = \det_{1 \le i, j \le N} p_{N-j}(x_{L,i})
 \, , \qquad
 \Delta_N(X_R) = \det_{1 \le i, j \le N} q_{N-j}(x_{R,i})
 \, ,
 \label{eq:Vandermonde_pq}
\end{align}
the partition function~\eqref{eq:Z_fn} is evaluated as a rank $N$ determinant,
\begin{align}
Z_N & = \frac{1}{N!^2} \int \prod_{k = L, R} \dd{X}_k \ee^{- \tr V_k(X_k)} \det_{1 \le i, j \le N} p_{N-j}(x_{L,i}) \det_{1 \le i, j \le N} \omega(x_{L,i},x_{R,j}) \det_{1 \le i, j \le N} q_{N-j}(x_{R,i})
    \nonumber \\
    & = \det_{1 \le i, j \le N} \qty[ \int \prod_{k = L, R} \dd{x}_k \ee^{-V_k(x_k)} p_{N-i}(x_L) \omega(x_L,x_R) q_{N-j}(x_R) ]
    \nonumber \\
    & = \det_{1 \le i, j \le N} \mathsf{N}_{N-i,N-j}
    \, ,
\end{align}
where we apply the AH identity for $X_{L,R}$.
This completes the proof.
\end{proof}

\begin{remark}[Biorthogonal polynomial]
Specializing the monic polynomials to the biorthogonal polynomials,
\begin{align}
    \cvev{ P_{i} \mid \omega \mid Q_{j} } = h_{i} \delta_{i,j}
    \, ,
\end{align}
the norm matrix is diagonalized $\mathsf{N}_{i,j} = h_i \delta_{i,j}$, so that the partition function is given by
\begin{align}
    Z_N = \prod_{i=0}^{N-1} h_i
    \, .
\end{align}
\end{remark}

\subsection{Christoffel--Darboux kernel}

\begin{definition}[Christoffel--Darboux kernel]
We define the Christoffel--Darboux (CD) kernel associated with the coupled matrix model,
\begin{align}
    K_N(x_R,x_L) & = \ee^{- V_L(x_L) - V_R(x_R)} \sum_{i,j = 0}^{N-1} q_i(x_R) \left(\mathsf{N}^{-1}\right)_{i,j} p_j(x_L)
    \nonumber \\
    & = \ee^{- V_L(x_L) - V_R(x_R)} \sum_{i = 0}^{N-1} \frac{Q_i(x_R) P_i(x_L)}{h_i}
    = \sum_{i = 0}^{N-1} \psi_i(x_R) \phi_i(x_L)
    \, .
    \label{eq:CD_def}
\end{align}
We denote the inverse of the norm matrix by $\left(\mathsf{N}^{-1}\right)_{i,j}$, and define the biorthonormal functions, that we call the wave functions, by
\begin{align}
    \phi_i(x) = \frac{\ee^{-V_L(x)}}{\sqrt{h_i}} p_i(x)
    \, , \qquad
    \psi_i(x) = \frac{\ee^{-V_R(x)}}{\sqrt{h_i}} q_i(x)
    \, .
    \label{eq:wf_def}
\end{align}
\end{definition}

\begin{proposition}
The probability distribution associated with the partition function~\eqref{eq:Z_fn} is written using the CD kernel,
\begin{align}
    \mathsf{P}_N(X_{L,R}) & = \frac{Z_N^{-1}}{N!^2} \prod_{k = L,R} \ee^{- \tr V_k(X_k)} \Delta_N(X_L) \det_{1 \le i, j \le N} \omega(x_{L,i},x_{R,j}) \Delta_N(X_R)
    \nonumber \\
    & = \frac{1}{N!^2} \det_{1 \le i, j \le N} \omega(x_{L,i},x_{R,j}) \det_{1 \le i, j \le N} K_N(x_{R,i},x_{L,j})    
    \, ,
    \label{eq:P_N}
\end{align}
which obeys the normalization condition
\begin{align}
    \int \prod_{k = L,R} \dd{X}_k \mathsf{P}_N(X_{L,R}) = 1
    \, .
\end{align}
\end{proposition}

\begin{definition}[Expectation value]
We define the expectation value with respect to the probability distribution function $\mathsf{P}_N(X_{L,R})$ as follows,
\begin{align}
    \rvev{ \mathcal{O}(X_{L,R}) } = \int \prod_{k = L,R} \dd{X}_k \mathsf{P}_N(X_{L,R}) \mathcal{O}(X_{L,R}) 
    \, .
\end{align}
\end{definition}

\subsection{Operator formalism}

\begin{definition}
We define an inner product symbol,
\begin{subequations}
\begin{align}
    \rvev{ f \mid g } & = \int \dd{x} f(x) g(x)
    \, , \\
    \rvev{ f \mid \omega \mid g } & = \int \dd{x}_{L,R} f(x_L) \omega(x_L,x_R) g(x_R)
    \, .
\end{align}
\end{subequations}
We remark that, compared with the previous notation~\eqref{eq:cvev}, this definition does not depend on the potential function.
\end{definition}

Then, the orthonormality of the wave functions $(\phi_i,\psi_i)$ defined in \eqref{eq:wf_def} is expressed as
\begin{align}
    \rvev{ \phi_i \mid \omega \mid \psi_j } 
    = \int \dd{x}_{L,R} \phi_i(x_L) \omega(x_L,x_R) \psi_j(x_R)
    = \delta_{i,j}
    \, ,
\end{align}
where we write
\begin{align}
    \phi_i(x) = \rvev{ \phi_i \mid x }
    \, , \qquad
    \psi_i(x) = \rvev{ x \mid \psi_i }
    \, , \qquad
    \omega(x_L,x_R) = \rvev{ x_L \mid \hat{\omega} \mid x_R }
    \, .
\end{align}
together with the completeness condition
\begin{align}
    1 = \int \dd{x} \ket{x} \bra{x}
    \, .
\end{align}

In this operator formalism, the CD kernel is given by a matrix element of the operator defined as
\begin{align}
    K_N(x_R, x_L) & = \rvev{ x_R \mid \hat{K}_N \mid x_L }
    \, , \qquad
    \hat{K}_N = \sum_{i = 0}^{N-1} \ket{\psi_i} \bra{\phi_i}
    \, .
\end{align}
Introducing infinite dimensional vectors
\begin{align}
    \ket{ \underline{\phi} } = 
    \begin{pmatrix}
    & \ket{\phi_0} & \ket{\phi_1} & \ket{\phi_2} & \cdots &
    \end{pmatrix}^\text{T}
    \, , \qquad
    \ket{ \underline{\psi} } = 
    \begin{pmatrix}
    & \ket{\psi_0} & \ket{\psi_1} & \ket{\psi_2} & \cdots &
    \end{pmatrix}^\text{T}
    \, ,
\end{align}
together with the projection matrix
\begin{align}
    \left( \Pi_N \right)_{i,j} =
    \begin{cases}
        1 & (i = j \in [0,\ldots,N-1]) \\
        0 & (\text{otherwise})
    \end{cases}
\end{align}
the CD kernel operator is written as
\begin{align}
    \hat{K}_N = \ket{ \underline{\psi} } \Pi_N \bra{ \underline{\phi} }
    \, .
\end{align}
In the limit $N \to \infty$, we have
\begin{align}
    \lim_{N \to \infty} {K}_N(x_R,x_L) = \sum_{i=0}^\infty \psi_i(x_R) \phi_i(x_L) = \rvev{x_R \mid \omega^{-1} \mid x_L } =: \tilde{\omega}(x_R,x_L)
    \, ,
\end{align}
such that
\begin{align}
    \int \dd{z} \omega(x,z) \tilde{\omega}(z,y) = \int \dd{z} \tilde\omega(x,z) {\omega}(z,y) = \delta(x-y)
    \, .
\end{align}

\begin{proposition}
The CD kernel is self-reproducing
\begin{align}
    \hat{K}_N \cdot \hat{\omega} \cdot \hat{K}_N = \hat{K}_N
    \, , \qquad
    \tr \left( \hat{\omega} \cdot \hat{K}_N \right) = N
    \, ,
\end{align}
and therefore the correlation functions are in general determinantal (Eynard--Mehta's theorem~\cite{Eynard:1998JPA}).
\end{proposition}

\subsection{Polynomial ensemble}\label{subsec:pe}

We consider a generalization of the coupled matrix model, that we call the coupled polynomial ensemble, which is a coupled version of the polynomial ensemble introduced in Ref.~\cite{Kuijlaars:2014RMTA}.
We define the following generalized coupled matrix model partition functions.
\begin{definition}
Let $(f_{k,i})_{k=L,R,0 = 1,\ldots,N-1}$ be a set of arbitrary functions.
We define the polynomial ensemble partition functions as follows,
\begin{subequations}\label{eq:Z_fn_pe}
\begin{align}
    Z_{N,f_L} & = \frac{1}{N!^2} \int \dd{X}_{L,R} \ee^{-\tr V_R(X_R)} \det_{1 \le i, j \le N} f_{L,N-i}(x_{L,j}) \det_{1 \le i, j \le N} \omega(x_{L,i},x_{R,j}) \Delta_N(X_R)
    \, , \\
    Z_{N,f_R} & = \frac{1}{N!^2} \int \dd{X}_{L,R} \ee^{-\tr V_L(X_L)} \Delta_N(X_L) \det_{1 \le i, j \le N} \omega(x_{L,i},x_{R,j}) \det_{1 \le i, j \le N} f_{R,N-i}(x_{R,j}) 
    \, .
\end{align}
\end{subequations}
\end{definition}
\begin{remark}
Specializing each function $(f_{k,i})_{k = L,R, i = 0,\ldots,N-1}$ to be a monic polynomial, these partition functions~\eqref{eq:Z_fn_pe} are reduced to the original one~\eqref{eq:Z_fn}.
\end{remark}

\noindent
These partition functions show the determinantal structure as discussed before.
In order to discuss their properties, we introduce the notation.
\begin{definition}[Mixed braket notation]
We define the following inner product symbol,
\begin{subequations}
\begin{align}
    \crvev{ f \mid g } & = \int \dd{x}_{L,R} \ee^{- V_L(x_L)} f(x_L) g(x_R)
    \, , \\
    \rcvev{ f \mid g } & = \int \dd{x}_{L,R} \ee^{- V_R(x_R)} f(x_L) g(x_R)
    \, .
\end{align}
\end{subequations}
\end{definition}

\noindent
We obtain the following result.

\begin{proposition}
The partition function of the polynomial ensemble is written as a rank $N$ determinant with a set of arbitrary monic polynomials $(p_i,q_i)_{i = 0,\ldots,N-1}$,
\begin{align}
    Z_{N,f_L} & = \det_{1 \le i, j \le N} \rcvev{ f_{L,N-i} \mid \omega \mid q_{N-j} }
    \, , \\
    Z_{N,f_R} & = \det_{1 \le i, j \le N} \crvev{ p_{N-i} \mid \omega \mid f_{R,N-j} }
    \, .
\end{align}
\end{proposition}
\begin{proof}
We obtain this formula by direct calculation.
Recalling the Vandermonde determinant is given as~\eqref{eq:Vandermonde_pq} with a set of arbitrary monic polynomials, we have
\begin{align}
    Z_{N,f_L} & = \frac{1}{N!^2} \int \dd{X}_{L,R} \ee^{-\tr V_R(X_R)} \det_{1 \le i, j \le N} f_{L,N-i}(x_{L,j}) \det_{1 \le i, j \le N} \omega(x_{L,i},x_{R,j}) \det_{1 \le i, j \le N} q_{N-j}(x_{R,i})
    \nonumber \\
    & = \det_{1 \le i, j \le N} \qty( \int \dd{x}_{L,R} \ee^{-V_R(x_R)} f_{L,N-i}(x_L) \omega(x_L,x_R) q_{N-j}(x_R) )
    \nonumber \\
    & = \det_{1 \le i, j \le N} \rcvev{ f_{L,N-i} \mid \omega \mid q_{N-j} }
    \, .
\end{align}
We can obtain the other formula in the same way.
\end{proof}

\begin{definition}[Biorthogonal functions]
We can then define two pairs of biorthogonal families $(F_{L,i},Q_{j})_{i,j = 0,\ldots,N-1}$ and $(P_{i}, F_{R,j})_{i,j = 0,\ldots,N-1}$ such that:
\begin{itemize}
\item The functions $P_i$ and $Q_j$ are monic polynomials. 
\item The functions $F_{L,i}$ (resp. $F_{R,i}$) are linearly spanned by the functions $(f_{L,k})_{k=0,\cdots,i}$ (resp. $(f_{R,k})_{k=0,\cdots,i}$). 
\item They satisfy the following scalar product properties: 
\begin{subequations}
\begin{align}
    \rcvev{ F_{L,i} \mid \omega \mid Q_{j} } &= h_{L,i} \delta_{i,j}
    \qquad (i,j=0,\ldots,N-1) , \\
    \crvev{ P_{i} \mid \omega \mid F_{R,j} } &= h_{R,i} \delta_{i,j}
    \qquad (i,j=0,\ldots,N-1) .
\end{align}
\end{subequations}
\end{itemize}
\end{definition}

\begin{corollary}
The partition functions of the coupled polynomial ensemble~\eqref{eq:Z_fn_pe} take the following form in terms of the normalization constants $(h_{k,i})_{k = L,R, i = 0,\ldots,N-1}$,
\begin{align}
    Z_{N,f_k} & = \prod_{i=0}^{N-1} h_{k,i}
    \qquad ( k = L, R )
    \, .
\end{align}
\end{corollary}
\begin{proof}
Once recalling that the determinant is invariant under linear operations on rows and columns, one can express it in terms of the biorthogonal functions defined before,
\begin{subequations}
\begin{align}
    Z_{N,f_L} & = \det_{1 \le i, j \le N} \rcvev{ F_{L,i-1} \mid \omega \mid Q_{j-1} }
    \, , \\
    Z_{N,f_R} & = \det_{1 \le i, j \le N} \crvev{ P_{i-1} \mid \omega \mid F_{R,j-1} }
    \, .
\end{align}
\end{subequations}
which is exactly the desired expression.
\end{proof}

\begin{definition}[Christoffel--Darboux kernel]
We define the CD kernels for the coupled polynomial ensemble as follows,
\begin{subequations}
\begin{align}
K_{N,f_L}(x,y) &= \ee^{-V_R(x)} \sum_{i=0}^{N-1} \frac{Q_i(x)F_{L,i}(y)}{h_{L,i}}  
\, , \\
K_{N,f_R}(x,y) &= \ee^{-V_L(y)} \sum_{i=0}^{N-1} \frac{F_{R,i}(x)P_i(y)}{h_{R,i}}
\, .
\end{align}
\end{subequations}
\end{definition}

\begin{remark}
As for the ordinary coupled matrix model~\eqref{eq:Z_fn}, one can define the following biorthonormal wave functions
\begin{align}
\psi_{L,i}(x) = \frac{1}{\sqrt{h_{L,i}}} \ee^{-V_R(x)} Q_i(x)
\, , \qquad
\phi_{L,i}(x) = \frac{1}{\sqrt{h_{L,i}}} F_{L,i}(x) \, ,\\
\psi_{R,i}(x) = \frac{1}{\sqrt{h_{R,i}}} \ee^{-V_L(x)} P_i(x)
\, , \qquad
\phi_{R,i}(x) = \frac{1}{\sqrt{h_{R,i}}} F_{R,i}(x)
\, ,
\end{align}
and the CD kernels take then a very compact form. 
\end{remark}

\begin{proposition}
The probability distributions for the coupled polynomial ensemble can be expressed as
\begin{align}
    \mathsf{P}_{N,f_L}(X_{L,R})  &= \frac{1}{N!^2} \det_{1 \le i, j \le N} \omega(x_{L,i},x_{R,j}) \det_{1 \le i, j \le N} K_{N,f_L}(x_{R,i},x_{L,j})    
    \, , \\
    \mathsf{P}_{N,f_R}(X_{L,R})  &= \frac{1}{N!^2} \det_{1 \le i, j \le N} \omega(x_{L,i},x_{R,j}) \det_{1 \le i, j \le N} K_{N,f_R}(x_{R,i},x_{L,j})
    \,.
\end{align}
\end{proposition}

\begin{remark}
All the previous formulas lead to the familiar matrix model formalism. Therefore the correlation functions of the Schur polynomial and the characteristic polynomials shown in the following sections are straightforwardly generalized to the coupled polynomial ensemble (except for the pair correlation functions discussed in Section~\ref{sec:pair_corr}).
We obtain a natural generalization of the results for the characteristic polynomial average with the source term~\cite{Kimura:2014mua,Kimura:2014vra,Kimura:2021hph} and also the one-matrix polynomial ensemble~\cite{Akemann:2020yti}.
\end{remark}

\section{Characteristic polynomial averages}\label{sec:ch_poly}

\subsection{Schur polynomial average}\label{sec:Schur_av}

We first compute the Schur polynomial average for the coupled matrix model, which will be a building block of the correlation functions of the characteristic polynomials~\cite{Kimura:2021hph}.
See also~\cite{Santilli:2021iks} for a related work.

\begin{definition}[Schur polynomial]
Let $\lambda$ be a partition, a non-increasing sequence of non-negative integers,
\begin{align}
    \lambda = (\lambda_1 \ge \lambda_2 \ge \cdots \ge \lambda_\ell > \lambda_{\ell+1} = \cdots = 0 )
    \, ,
\end{align}
where $\ell = \ell(\lambda)$ is called the length of the partition.
Denoting the transposed partition by $\lambda^\text{T}$, we have $\ell(\lambda) = \lambda_1^\text{T}$.
Then, the Schur polynomial of $N$ variables, $X = (x_i)_{i = 1,\ldots,N}$, is defined as follows,
\begin{align}
    s_\lambda(X) = \frac{1}{\Delta_N(X)} \det_{1 \le i, j \le N} x_i^{\lambda_j + N - i}
    \, .
    \label{eq:Schur_def}
\end{align}
If $\ell(\lambda) > N$, we have $s_\lambda(X) = 0$.
We also remark $s_\emptyset(X) = 1$.
\end{definition}

\begin{lemma}\label{lemma:Schur_average}
The Schur polynomial average with respect to the probability distribution function $\mathsf{P}_N(X_{L,R})$~\eqref{eq:P_N} is given as a rank $N$ determinant,
\begin{align}
    \rvev{ s_{\lambda}(X_L) s_\mu(X_R) }
    & = \frac{1}{Z_N} \det_{1 \le i, j \le N} \cvev{ x_{L}^{\lambda_{i} + N - i} \mid \omega \mid x_R^{\mu_j + N - j} }
    \, .
\end{align}
\end{lemma}
\begin{proof}
This can be shown by direct calculation,
\begin{align}
    & \rvev{ s_{\lambda}(X_L) s_\mu(X_R) }
    \nonumber \\
    & = \int \prod_{k = L,R} \dd{X}_k \mathsf{P}_N(X_{L,R}) s_{\lambda}(X_L) s_\mu(X_R)
    \nonumber \\
    & = \frac{Z_N^{-1}}{N!^2} \int \prod_{k = L,R} \ee^{- \tr V_k(X_k)} \det_{1 \le i, j \le N} x_{L,i}^{\lambda_j + N - j} \det_{1 \le i, j \le N} \omega(x_{L,i},x_{R,j}) \det_{1 \le i, j \le N} x_{R,i}^{\mu_j + N - j}
    \nonumber \\
    & = \frac{1}{Z_N} \det_{1 \le i, j \le N} \qty( \int \prod_{k = L,R} \dd{x}_{k} \ee^{-V_k(x_k)} x_{L}^{\lambda_i + N - i} \omega(x_L,x_R) x_R^{\mu_j + N - j} )
    \nonumber \\
    & = \frac{1}{Z_N} \det_{1 \le i, j \le N} \cvev{ x_{L}^{\lambda_{i} + N - i} \mid \omega \mid x_R^{\mu_j + N - j} }
    \, .
\end{align} 
This completes the proof.
\end{proof}

\begin{lemma}[Schur polynomial expansion]\label{lemma:Schur_expansion}
Let $Z = \diag(z_1,\ldots,z_M)$.
The characteristic polynomial is expanded with the Schur polynomial as follows,
\begin{subequations}
\begin{align}
    \prod_{\alpha = 1}^M \det(z_\alpha - X)
    & = \sum_{\lambda \subseteq (M^N)} (-1)^{|\lambda|} s_{\lambda^\vee} (Z) s_\lambda(X)
    \, , \\
    \prod_{\alpha = 1}^M \det(z_\alpha - X)^{-1}
    & = \det_M Z^{-N} \sum_{\lambda | \ell(\lambda) \le \operatorname{min}(M,N)} s_\lambda(Z^{-1}) s_\lambda(X)
    \, ,
\end{align}
\end{subequations}
where we define the dual partition
\begin{align}
    \lambda^\vee = (\lambda_1^\vee,\ldots,\lambda_M^\vee) = (N - \lambda_M^\text{T}, \ldots, N - \lambda_1^\text{T})
    \, ,
\end{align}
and the length of the partition denoted by $\ell(\lambda) = \lambda_1$.
\end{lemma}
\begin{proof}
This follows from the Cauchy sum formula.
See, e.g.,~\cite{Macdonald:2015}.
\end{proof}

\subsection{Characteristic polynomial}\label{sec:ch_poly_av}

Based on the Schur polynomial expansion, we obtain the determinantal formula for the characteristic polynomial average as follows.
\begin{proposition}[Characteristic polynomial average]
The $M$-point correlation function of the characteristic polynomial is given by a rank $M$ determinant of the associated biorthogonal polynomials,
\begin{subequations}
\begin{align}
    \expval{ \prod_{\alpha = 1}^M \det(z_\alpha - X_L) }
    & = 
    \frac{1}{\Delta_M(Z)} \det_{1 \le \alpha, \beta \le M }  P_{N+M-\beta}(z_\alpha) 
    \, , \label{eq:ch_poly_av_L} \\
    \expval{ \prod_{\alpha = 1}^M \det(z_\alpha - X_R) }
    & = 
    \frac{1}{\Delta_M(Z)} \det_{1 \le \alpha, \beta \le M } Q_{N+M-\beta}(z_\alpha)     
    \, .
    \label{eq:ch_poly_av_R}
\end{align}
\end{subequations}
\end{proposition}
\begin{proof}
We may use Lemma~\ref{lemma:Schur_average} and Lemma~\ref{lemma:Schur_expansion} to show this formula.
Considering the characteristic polynomial coupled with the matrix $X_{L}$, we obtain
\begin{align}
    \expval{ \prod_{\alpha = 1}^M \det(z_\alpha - X_L) }
    & = \sum_{\lambda \subseteq (M^N)} (-1)^{|\lambda|} s_{\lambda^\vee} (Z) \expval{s_\lambda(X_L)}
    \nonumber \\
    & = \frac{Z_N^{-1}}{\Delta_M(Z)} \sum_{\lambda \subseteq (M^N)} (-1)^{|\lambda|} \det_{1 \le \alpha, \beta \le M} z_\alpha^{\lambda_\beta^\vee + M - \beta}  \det_{1 \le i,j \le N} ( x_L^{\lambda_i + N - i} \mid \omega \mid q_{N-j} )
    \nonumber \\
    & = \frac{Z_N^{-1}}{\Delta_M(Z)} \det_{\substack{1 \le \alpha, \beta \le M \\ 1 \le i, j \le N}} 
    \begin{pmatrix}
    z_\alpha^{N+M-\beta} & z_\alpha^{N-j}\\
    ( x_L^{N+M-\beta} \mid \omega \mid q_{N-j} ) & ( x_L^{N-j} \mid \omega \mid q_{N-j} )
    \end{pmatrix}
    \nonumber \\
    & = \frac{Z_N^{-1}}{\Delta_M(Z)} \det_{\substack{1 \le \alpha, \beta \le M \\ 1 \le i, j \le N}} 
    \begin{pmatrix}
    p_{N+M-\beta}(z_\alpha) & p_{N-j}(z_\alpha) \\
    ( p_{N+M-\beta} \mid \omega \mid q_{N-i} ) & ( p_{N-j} \mid \omega \mid q_{N-i} )
    \end{pmatrix}
    \nonumber \\
    & = \frac{Z_N^{-1}}{\Delta_M(Z)} \det_{\substack{1 \le \alpha, \beta \le M \\ 1 \le i, j \le N}} 
    \begin{pmatrix}
    P_{N+M-\beta}(z_\alpha) & P_{N-j}(z_\alpha) \\
    ( P_{N+M-\beta} \mid \omega \mid Q_{N-i} ) & ( P_{N-j} \mid \omega \mid Q_{N-i} )
    \end{pmatrix}    
    \nonumber \\
    & = \frac{Z_N^{-1}}{\Delta_M(Z)} \det_{\substack{1 \le \alpha, \beta \le M \\ 1 \le i, j \le N}} 
    \begin{pmatrix}
    P_{N+M-\beta}(z_\alpha) & P_{N-j}(z_\alpha) \\
    0 & h_{N-i} \, \delta_{N-i,N-j}
    \end{pmatrix}   
    \nonumber \\
    & = \frac{1}{\Delta_M(Z)} \det_{1 \le \alpha, \beta \le M }  P_{N+M-\beta}(z_\alpha) 
    \, ,
\end{align}
where we apply the rank $M$ co-factor expansion of the rank $N+M$ determinant.
The other formula~\eqref{eq:ch_poly_av_R} is similarly obtained.
\end{proof}

\subsection{Characteristic polynomial inverse}\label{sec:ch_poly_inv_av}

In order to write down the characteristic polynomial inverse average, we define the Hilbert transform.
\begin{definition}[Hilbert transform]
We define the Hilbert transform of the polynomial functions as follows,
\begin{subequations}
\begin{align}
    \widetilde{p}_{j}(z)
    = \int \prod_{k=L,R} \dd{x}_{k} \ee^{-V_k(x_k)} \frac{\omega(x_L,x_R) q_{j}(x_R) }{z - x_L}
    \, , \\
    \widetilde{q}_{j}(z)
    = \int \prod_{k=L,R} \dd{x}_{k} \ee^{-V_k(x_k)} \frac{p_{j}(x_L) \omega(x_L,x_R)}{z - x_R}
    \, .
\end{align}
\end{subequations}
\end{definition}

We obtain the following formula.
\begin{proposition}[Characteristic polynomial inverse average]
Let $Z = \diag(z_1,\ldots,z_M)$.
The $M$-point correlation function of the characteristic polynomial inverse is given by a rank $M$ determinant of the dual biorthogonal polynomials.
Depending on the relation between $N$ and $M$, we have the following formulas.
\begin{enumerate}
    \item $M \le N$
    \begin{subequations}
    \begin{align}
    \expval{ \prod_{\alpha = 1}^M \det(z_\alpha - X_L)^{-1} }
    & = \frac{Z_{N-M}/Z_N}{\Delta_M(Z)} 
    \det_{1 \le \alpha, \beta \le M} \widetilde{P}_{N-\beta}(z_\alpha)
    \, , \label{eq:ch_poly_av_inv_L} \\[.5em]
    \expval{ \prod_{\alpha = 1}^M \det(z_\alpha - X_R)^{-1} }
    & = \frac{Z_{N-M}/Z_N}{\Delta_M(Z)} 
    \det_{1 \le \alpha, \beta \le M} \widetilde{Q}_{N-\beta}(z_\alpha)    
    \, .
    \label{eq:ch_poly_av_inv_R}
    \end{align}
    \item $M \ge N$
    \begin{align}
    \expval{ \prod_{\alpha = 1}^M \det(z_\alpha - X_L)^{-1} }
    & = \frac{Z_N^{-1}}{\Delta_N(Z)}
    \det_{\substack{i=1,\ldots,N \\ \alpha = 1,\ldots,M \\ a = 1,\ldots,M-N}}
    \begin{pmatrix}
    \widetilde{p}_{N-i}(z_\alpha) \\ p_{a-1}(z_\alpha)
    \end{pmatrix}    
    \, , \label{eq:ch_poly_av_inv2_L} \\[.5em]
    \expval{ \prod_{\alpha = 1}^M \det(z_\alpha - X_R)^{-1} }
    & = \frac{Z_N^{-1}}{\Delta_N(Z)}
    \det_{\substack{i=1,\ldots,N \\ \alpha = 1,\ldots,M \\ a = 1,\ldots,M-N}}
    \begin{pmatrix}
    \widetilde{q}_{N-i}(z_\alpha) \\ q_{a-1}(z_\alpha)
    \end{pmatrix}    
    \, .
    \label{eq:ch_poly_av_inv2_R}
    \end{align}
    \end{subequations}
\end{enumerate}
\end{proposition}
\begin{proof}
We first consider the case $M \le N$.
In this case, the Schur polynomial average for $\ell(\lambda) \le M$ is obtained from Lemma~\ref{lemma:Schur_average} as
\begin{align}
    \expval{s_\lambda(X_L)}
    & = \frac{1}{Z_N}
    \det_{\substack{1 \le \alpha, \beta \le M \\ M+1 \le a, b \le N}} 
    \begin{pmatrix}
    ( x_L^{\lambda_\alpha + N - \alpha} \mid \omega \mid q_{N-\beta} ) & ( x_L^{N - a} \mid \omega \mid q_{N-\beta} ) \\[.5em]
    ( x_L^{\lambda_\alpha + N - \alpha} \mid \omega \mid q_{N-b} ) & ( x_L^{N - a} \mid \omega \mid q_{N-b} ) \\
    \end{pmatrix}
    \nonumber \\
    & = \frac{1}{Z_N}
    \det_{\substack{1 \le \alpha, \beta \le M \\ M+1 \le a, b \le N}} 
    \begin{pmatrix}
    ( x_L^{\lambda_\alpha + N - \alpha} \mid \omega \mid Q_{N-\beta} ) & 0 \\[.5em]
    ( x_L^{\lambda_\alpha + N - \alpha} \mid \omega \mid Q_{N-b} ) & h_{N-a} \, \delta_{N-a, N-b} \\
    \end{pmatrix}   
    \nonumber \\
    & = \frac{Z_{N-M}}{Z_N} \det_{1 \le \alpha, \beta \le M} ( x_L^{\lambda_\alpha + N - \alpha} \mid \omega \mid Q_{N-\beta} )
    \, .
\end{align}
Then, applying the Schur polynomial expansion as given in Lemma~\ref{lemma:Schur_expansion}, the characteristic polynomial inverse average is given as follows,
\begin{align}
    & \expval{ \prod_{\alpha = 1}^M \det(z_\alpha - X_L)^{-1} }
    \nonumber \\
    & = \det_M Z^{-N} \sum_{\ell(\lambda) \le M} s_{\lambda} (Z^{-1}) \expval{s_\lambda(X_L)}
    \nonumber \\    
    & = \frac{Z_{N-M}}{Z_N} \frac{1}{\Delta_M(Z)}
    \sum_{0 \le \lambda_M \le \cdots \le \lambda_1 \le \infty} 
    \det_{1 \le \alpha, \beta \le M} \qty( z_\alpha^{-\lambda_\beta + \beta - (N+1)} )
    \det_{1 \le \alpha, \beta \le M} ( x_L^{\lambda_\alpha + N - \alpha} \mid \omega \mid Q_{N-\beta} )
    \nonumber \\    
    & = \frac{Z_{N-M}}{Z_N} \frac{1}{\Delta_M(Z)} \frac{1}{M!}
    \sum_{\substack{0 \le r_1, \cdots, r_M \le \infty \\ r_\alpha \neq r_\beta}} 
    \det_{1 \le \alpha, \beta \le M} \qty( z_\alpha^{M - N - r_\beta - 1} )
    \det_{1 \le \alpha, \beta \le M} ( x_L^{N - M + r_\alpha} \mid \omega \mid Q_{N-\beta} )
    \nonumber \\    
    & = \frac{Z_{N-M}}{Z_N} \frac{1}{\Delta_M(Z)} 
    \det_{1 \le \alpha, \beta \le M} \qty( \sum_{r=0}^\infty z_\alpha^{M-N-r-1} ( x_L^{N - M + r} \mid \omega \mid Q_{N-\beta} ) )
    \, ,
\end{align}
where we have applied an analog of the AH identity for non-colliding discrete variables, $(r_\alpha)_{\alpha = 1,\ldots,M}$ ($r_\alpha \neq r_\beta$).
Noticing
\begin{align}
    \sum_{r=0}^\infty z^{-r-1} x^r = \frac{1}{z-x}
    \, ,
\end{align}
and
\begin{align}
    \frac{x^{N-M}}{z-x} 
    = \frac{z^{N-M}}{z - x} - \frac{z^{N-M} - x^{N-M}}{z - x} 
    = \frac{z^{N-M}}{z - x} - O(x^{N-M-1})
    \, ,
\end{align}
we obtain 
\begin{align}
    & \sum_{r=0}^\infty z_\alpha^{M-N-r-1} ( x_L^{N - M + r} \mid \omega \mid Q_{N-\beta} )
    \nonumber \\
    & = 
    z_\alpha^{M-N} \int \prod_{k=L,R} \dd{x}_{k} \ee^{-V_k(x_k)} \frac{x_L^{N-M}}{z_\alpha - x_L} \omega(x_L,x_R) Q_{N-\beta}(x_R)
    \nonumber \\
    & = 
    \int \prod_{k=L,R} \dd{x}_{k} \ee^{-V_k(x_k)} \frac{\omega(x_L,x_R) Q_{N-\beta}(x_R) }{z_\alpha - x_L} 
    \nonumber \\
    & = \widetilde{P}_{N-\beta}(z_\alpha) 
    \, .
\end{align}
We have used the biorthogonality $\cvev{ x_L^{a} \mid \omega \mid Q_{N - \beta}} = 0$ for $\beta = 1, \ldots,M$ and $a = 0,\ldots,N-M-1$ to obtain the last expression. 
This completes the derivation of the formula \eqref{eq:ch_poly_av_inv_L}.
We can similarly obtain the formula \eqref{eq:ch_poly_av_inv_R}.

We then consider the case $M \ge N$. 
In this case, the $M$-variable Schur polynomial with the condition $\ell(\lambda) \le N$ for $Z = \diag(z_1,\ldots,z_M)$ is given by
\begin{align}
    \frac{s_\lambda(Z^{-1})}{\det Z^N} 
    & = 
    \det_{\substack{i=1,\ldots,N \\ \alpha = 1,\ldots,M \\ a = 1,\ldots,M-N}}
    \begin{pmatrix}
    z_\alpha^{-\lambda_i + i - (N+1)} \\ z_\alpha^{a-1}
    \end{pmatrix}
    & = 
    \det_{\substack{i=1,\ldots,N \\ \alpha = 1,\ldots,M \\ a = 1,\ldots,M-N}}
    \begin{pmatrix}
    z_\alpha^{-\lambda_i + i - (N+1)} \\ p_{a-1}(z_\alpha)
    \end{pmatrix}
    \, .
\end{align}
Hence, applying the Schur polynomial expansion, we obtain
\begin{align}
    & \expval{ \prod_{\alpha = 1}^M \det(z_\alpha - X_L)^{-1} }
    \nonumber \\
    & = \frac{Z_N^{-1}}{\Delta_N(Z)}
    \sum_{0 \le \lambda_N \le \cdots \le \lambda_1 \le \infty}
    \det_{\substack{i=1,\ldots,N \\ \alpha = 1,\ldots,M \\ a = 1,\ldots,M-N}}
    \begin{pmatrix}
    z_\alpha^{-\lambda_i + i - (N+1)} \\ p_{a-1}(z_\alpha)
    \end{pmatrix}
    \det_{1 \le i,j \le N} ( x_L^{\mu_i + N - i} \mid \omega \mid q_{N - j} )
    \nonumber \\
    & = 
    \frac{Z_N^{-1}}{\Delta_N(Z)}
    \det_{\substack{i=1,\ldots,N \\ \alpha = 1,\ldots,M \\ a = 1,\ldots,M-N}}
    \begin{pmatrix}
    \displaystyle \sum_{r=0}^\infty
    z_\alpha^{-r-1} ( x_L^{r} \mid \omega \mid q_{N - i} ) \\ p_{a-1}(z_\alpha)
    \end{pmatrix}
    \nonumber \\
    & = 
    \frac{Z_N^{-1}}{\Delta_N(Z)}
    \det_{\substack{i=1,\ldots,N \\ \alpha = 1,\ldots,M \\ a = 1,\ldots,M-N}}
    \begin{pmatrix}
    \displaystyle \sum_{r=0}^\infty
    z_\alpha^{-r-1} ( x_L^{r} \mid \omega \mid q_{N - i} ) \\ p_{a-1}(z_\alpha)
    \end{pmatrix}
    \nonumber \\
    & = 
    \frac{Z_N^{-1}}{\Delta_N(Z)}
    \det_{\substack{i=1,\ldots,N \\ \alpha = 1,\ldots,M \\ a = 1,\ldots,M-N}}
    \begin{pmatrix}
    \widetilde{p}_{N-i}(z_\alpha) \\ p_{a-1}(z_\alpha)
    \end{pmatrix}    
    \, .
\end{align}
This is the determinantal formula shown in~\eqref{eq:ch_poly_av_inv2_L}.
We can similarly obtain the other formula \eqref{eq:ch_poly_av_inv2_R}.
This completes the proof.
\end{proof}

\newpage

\section{Pair correlation functions}\label{sec:pair_corr}

In this Section, we consider the correlation function of both of the characteristic polynomials coupled to the matrices $X_{L,R}$, that we call the pair correlation function.

\subsection{Characteristic polynomial}

We have the following result regarding the pair correlation of the characteristic polynomials.

\begin{proposition}[Pair correlation of characteristic polynomials]\label{thm:pair}
Let $Z = \diag(z_1,\ldots,z_M)$ and $W = \diag(w_1,\ldots,w_M)$.
The correlation function of $M$ pairs of the characteristic polynomials is given by a rank $M$ determinant of the CD kernel,
\begin{align}
    \expval{ \prod_{\alpha = 1}^M \det(z_\alpha - X_L) \det(w_\alpha - X_R) }
    & =
    \frac{\ee^{\tr V_L(Z) + \tr V_R(W)}}{\Delta_M(Z) \Delta_M(W)}
    \frac{Z_{N+M}}{Z_N}
    \det_{1 \le \alpha, \beta \le M} 
    K_{N+M}(w_\alpha,z_\beta) 
    \, .
\end{align}
\end{proposition}

\begin{proof}
We use Lemma~\ref{lemma:Schur_average} and Lemma~\ref{lemma:Schur_expansion} as before.
In addition, we apply the co-factor expansion twice to obtain the following,
\begin{align}
    & \expval{ \prod_{\alpha = 1}^M \det(z_\alpha - X_L) \det(w_\alpha - X_R) }
    \nonumber \\
    & =
    \frac{Z_N^{-1}}{\Delta_M(Z) \Delta_M(W)}
    \det_{ \substack{1 \le \alpha, \beta \le M \\ 1 \le i, j \le N + M} }
    \begin{pmatrix}
    0 & w_\alpha^{N + M - j} \\
    z_\beta^{N + M - i} & ( x_L^{N + M - i} \mid \omega \mid x_R^{N + M - j} )
    \end{pmatrix}
    \nonumber \\
    & =
    \frac{Z_N^{-1}}{\Delta_M(Z) \Delta_M(W)}
    \det_{ \substack{1 \le \alpha, \beta \le M \\ 1 \le i, j \le N + M} }
    \begin{pmatrix}
    0 & q_{N + M - j}(w_\alpha) \\
    p_{N + M - i}(z_\beta) & \mathsf{N}_{N+M-i,N+M-j}
    \end{pmatrix}    
    \nonumber \\
    & =
    \frac{Z_{N+M}/Z_N}{\Delta_M(Z) \Delta_M(W)}
    \det_{1 \le \alpha, \beta \le M} 
    \qty(
    \sum_{k,k'=0}^{N+M-1} q_{k}(w_\alpha) (\mathsf{N}^{-1})_{k,k'} p_{k'}(z_\beta)
    )
    \nonumber \\
    & =
    \frac{\ee^{\tr V_L(Z) + \tr V_R(W)}}{\Delta_M(Z) \Delta_M(W)}
    \frac{Z_{N+M}}{Z_N}
    \det_{1 \le \alpha, \beta \le M} 
    K_{N+M}(w_\alpha,z_\beta) 
    \, ,
\end{align}
We have applied the definition of the CD kernel of degree $N+M$~\eqref{eq:CD_def} to obtain the last expression.
\end{proof}
\begin{remark}
This result can be also obtained using the self-reproducing property of the CD kernel as follows.
Noticing
\begin{align}
    \Delta_N(X) \prod_{\alpha = 1}^M \det(z_\alpha - X) = \frac{\Delta_{N+M}(X;Z)}{\Delta_M(Z)}
    \, ,
\end{align}
the pair correlation is given by
\begin{align}
    & \expval{ \prod_{\alpha = 1}^M \det(z_\alpha - X_L) \det(w_\alpha - X_R) }
    \nonumber \\
    & = \frac{Z_N^{-1}}{\Delta_M(Z) \Delta_M(W)} \frac{1}{N!^2}
    \int \prod_{k=L,R} \dd{X}_k \ee^{-\tr V_k(X_k)} \Delta_{N+M}(X_L;Z) \det_{1 \le i, j \le N} \omega(x_{L,i},x_{R,j}) \Delta_{N+M}(X_R;W)
    \nonumber \\
    & = \frac{\ee^{\tr V_L(Z) + \tr V_R(W)}}{\Delta_M(Z) \Delta_M(W)} \frac{Z_{N+M}/Z_N}{N!^2}
    \nonumber \\
    & \hspace{1em} \times
    \int \prod_{k=L,R} \dd{X}_k \ee^{-\tr V_k(X_k)} \det_{1 \le i, j \le N} \omega(x_{L,i},x_{R,j}) \det_{\substack{1 \le i, j \le N \\ 1 \le \alpha, \beta \le M}} 
    \begin{pmatrix}
    K_{N+M}(x_{R,i},x_{L,j}) & K_{N+M}(x_{R,i},z_{\beta}) \\
    K_{N+M}(w_{\alpha},x_{L,j}) & K_{N+M}(w_{\alpha},z_\beta)
    \end{pmatrix}
    \nonumber \\
    & =
    \frac{\ee^{\tr V_L(Z) + \tr V_R(W)}}{\Delta_M(Z) \Delta_M(W)}
    \frac{Z_{N+M}}{Z_N}
    \det_{1 \le \alpha, \beta \le M} 
    K_{N+M}(w_\alpha,z_\beta) 
    \, .
\end{align}
\end{remark}

\subsection{Characteristic polynomial inverse}

We then consider the pair correlation of the characteristic polynomial inverses.
In order to write down the formula in this case, we define the dual CD kernel as follows.
\begin{definition}[Dual Christoffel--Darboux kernel]
For the dual wave functions defined through the Hilbert transform,
\begin{subequations}
\begin{align}
    \widetilde{\phi}_i(z) & = \ee^{V_L(z)} \int \dd{x}_{L,R} \ee^{-V_L(x_L)} \frac{\omega(x_L, x_R) \psi_i(x_R)}{z - x_L}
    \, , \\
    \widetilde{\psi}_i(z) & = \ee^{V_R(z)} \int \dd{x}_{L,R} \ee^{-V_R(x_R)} \frac{\phi_i(x_L)\omega(x_L, x_R) }{z - x_R}
    \, ,
\end{align}
\end{subequations}
we define the dual Christoffel--Darboux kernel of degree $N$ as follows,
\begin{align}
    \widetilde{K}_{N}(w,z) = \sum_{i = N}^\infty \widetilde{\psi}_i(w) \widetilde{\phi}_i(z)
    \, .
\end{align}
\end{definition}
\begin{proposition}[Pair correlation of characteristic polynomial inverses]
Let $Z = \diag(z_1,\ldots,z_M)$ and $W = \diag(w_1,\ldots,w_M)$.
The correlation function of $M$ pairs of the characteristic polynomial inverses is given by a rank $M$ determinant of the dual CD kernel depending on the relation between $N$ and $M$ as follows.
\begin{enumerate}
    \item $M \le N$
    \begin{subequations}
    \begin{align}
    \hspace{-1.5em}
    & \expval{ \prod_{\alpha = 1}^M \det(z_\alpha - X_L)^{-1} \det(w_\alpha - X_R)^{-1} }
    =
    \frac{\ee^{-\operatorname{tr}V_L(Z)}\ee^{-\operatorname{tr}V_R(W)}}{\Delta_M(Z) \Delta_M(W)} \frac{Z_{N-M}}{Z_N}
    \det_{1 \le \alpha, \beta \le M}
    \widetilde{K}_{N-M}(w_\beta,z_\alpha) 
    \label{eq:ch_poly_pair_inv1} 
    \end{align}
    \item $M \ge N$
    \begin{align}
    & \expval{ \prod_{\alpha = 1}^M \det(z_\alpha - X_L)^{-1} \det(w_\alpha - X_R)^{-1} }
    \nonumber \\
    & =
     \frac{(-1)^{M-N} Z_N^{-1}}{\Delta_M(Z) \Delta_M(W)}
    \det_{1 \le \alpha, \beta \le M} \qty( \frac{1}{z_\alpha - x_L} \mid \omega \mid \frac{1}{w_\beta - x_R})
    \det_{1 \le a, b \le M - N}
    \qty( \sum_{\alpha,\beta = 1}^M p_{a-1}(z_\alpha) \widetilde{\omega}_{\alpha,\beta} q_{b-1}(w_\beta)  )
    \label{eq:ch_poly_pair_inv2} 
    \end{align}    
    \end{subequations}
    where $\widetilde{\omega}_{\alpha,\beta}$ is the inverse of $\qty( \frac{1}{z_\alpha - x_L} \mid \omega \mid \frac{1}{w_\beta - x_R})$.
\end{enumerate}
\end{proposition}
\begin{proof}
We first consider the case $M \le N$.
In this case, applying the Schur polynomial expansion as before, we obtain
\begin{align}
    & \expval{ \prod_{\alpha = 1}^M \det(z_\alpha - X_L)^{-1} \det(w_\alpha - X_R)^{-1} }
    \nonumber \\
    & =    
    \frac{Z_N^{-1}}{\Delta_M(Z) \Delta_M(W)}
    \sum_{\ell(\lambda), \ell(\mu) \le M}
    \det_{1 \le \alpha, \beta \le M} z_\alpha^{- N - \lambda_\beta + \beta - 1}
    \det_{1 \le \alpha, \beta \le M} w_\alpha^{- N - \mu_\beta + \beta - 1}
    \nonumber \\
    & \hspace{8em}
    \times \det_{\substack{ 1 \le \alpha, \beta \le M \\ 1 \le i, j \le N - M}}
    \begin{pmatrix}
    (x_L^{\lambda_\alpha + N - \alpha} \mid \omega \mid x_R^{\mu_\beta + N - \beta} ) & (x_L^{\lambda_\alpha + N - \alpha} \mid \omega \mid x_R^{N - M - j} ) \\
    (x_L^{N - M - i} \mid \omega \mid x_R^{\mu_\beta + N - \beta} ) & (x_L^{N - M - i} \mid \omega \mid x_R^{N - M - j} )
    \end{pmatrix}
    \nonumber \\
    & =    
    \frac{Z_N^{-1}}{\Delta_M(Z) \Delta_M(W)}
    \sum_{\ell(\lambda), \ell(\mu) \le M}
    \det_{1 \le \alpha, \beta \le M} z_\alpha^{- N - \lambda_\beta + \beta - 1}
    \det_{1 \le \alpha, \beta \le M} w_\alpha^{- N - \mu_\beta + \beta - 1}
    \nonumber \\
    & \hspace{8em}
    \times \det_{\substack{ 1 \le \alpha, \beta \le M \\ 1 \le i, j \le N - M}}
    \begin{pmatrix}
    (x_L^{\lambda_\alpha + N - \alpha} \mid \omega \mid x_R^{\mu_\beta + N - \beta} ) & (x_L^{\lambda_\alpha + N - \alpha} \mid \omega \mid q_{N - M - j} ) \\
    (p_{N - M - i} \mid \omega \mid x_R^{\mu_\beta + N - \beta} ) & (p_{N - M - i} \mid \omega \mid q_{N - M - j} )
    \end{pmatrix}   
    \nonumber \\
    & =    
    \frac{Z_{N-M} / Z_N}{\Delta_M(Z) \Delta_M(W)}
    \sum_{\ell(\lambda), \ell(\mu) \le M}
    \det_{1 \le \alpha, \beta \le M} z_\alpha^{- N - \lambda_\beta + \beta - 1}
    \det_{1 \le \alpha, \beta \le M} w_\alpha^{- N - \mu_\beta + \beta - 1}
    \nonumber \\
    & \quad 
    \times \det_{ 1 \le \alpha, \beta \le M}
    \qty(
    (x_L^{\lambda_\alpha + N - \alpha} \mid \omega \mid x_R^{\mu_\beta + N - \beta} ) - \sum_{i,j=0}^{N - M - 1} 
    (x_L^{\lambda_\alpha + N - \alpha} \mid \omega \mid q_{i} ) (\mathsf{N}^{-1})_{i,j} (p_{j} \mid \omega \mid x_R^{\mu_\beta + N - \beta} )
    )     
    \, .
\end{align}
We remark that each element in the determinant is given by
\begin{align}
    &
    (x_L^{\lambda_\alpha + N - \alpha} \mid \omega \mid x_R^{\mu_\beta + N - \beta} ) - \sum_{i,j=0}^{N - M - 1} 
    (x_L^{\lambda_\alpha + N - \alpha} \mid \omega \mid q_{i} ) (\mathsf{N}^{-1})_{i,j} (p_{j} \mid \omega \mid x_R^{\mu_\beta + N - \beta} )
    \nonumber \\
    & = 
    (x_L^{\lambda_\alpha + N - \alpha} \mid \omega \mid x_R^{\mu_\beta + N - \beta} )
    \nonumber \\
    & \qquad 
    - \int \prod_{k=L,R,L',R'} \dd{x}_k \ee^{-V_k(x_k)} x_L^{\lambda_\alpha + N - \alpha} \omega(x_L,x_{R'}) \sum_{i,j=0}^{N-M-1} q_{i}(x_{R'}) (\mathsf{N}^{-1})_{i,j} p_{j}(x_{L'}) \omega(x_{L'},x_R) x_L^{\mu_\beta + N - \beta}
    \nonumber \\
    & = 
    (x_L^{\lambda_\alpha + N - \alpha} \mid \omega \mid x_R^{\mu_\beta + N - \beta} )
    \nonumber \\
    & \qquad 
    - \int \prod_{k=L,R} \dd{x}_k \dd{x}_{k'} \ee^{-V_k(x_k)} x_L^{\lambda_\alpha + N - \alpha} \omega(x_L,x_{R'}) K_{N-M}(x_{R'},x_{L'}) \omega(x_{L'},x_R) x_R^{\mu_\beta + N - \beta}    
    \nonumber \\
    & = 
    \int \prod_{k=L,R} \dd{x}_k \dd{x}_{k'} \ee^{-V_k(x_k)} x_L^{\lambda_\alpha + N - \alpha} \omega(x_L,x_{R'}) \qty( \widetilde{\omega}(x_{R'},x_{L'}) - K_{N-M}(x_{R'},x_{L'}) ) \omega(x_{L'},x_{R}) x_R^{\mu_\beta + N - \beta}
    \nonumber \\
    & = 
    \int \prod_{k=L,R} \dd{x}_k \dd{x}_{k'} \ee^{-V_k(x_k)} x_L^{\lambda_\alpha + N - \alpha} \omega(x_L,x_{R'}) \qty( \sum_{k=N-M}^\infty \psi_k(x_{R'}) \phi_k(x_{L'})  ) \omega(x_{L'},x_R) x_R^{\mu_\beta + N - \beta}
    \, .
\end{align}
Therefore, we obtain
\begin{align}
    & \expval{ \prod_{\alpha = 1}^M \det(z_\alpha - X_L)^{-1} \det(w_\alpha - X_R)^{-1} }
    \nonumber \\
    & =    
    \frac{Z_{N-M} / Z_N}{\Delta_M(Z) \Delta_M(W)}
    \det_{1 \le \alpha, \beta \le M}
    \qty(
    \sum_{i=N-M}^\infty
    \int \prod_{k=L,R} \dd{x}_k \dd{x}_{k'} \ee^{-V_k(x_k)}  
    \frac{\omega(x_L,x_{R'}) \psi_i(x_{R'}) \phi_i(x_{L'}) \omega(x_{L'},x_R) }{(z_\alpha - x_{L}) (w_\beta - x_{R})}
    )
    \nonumber \\
    & =    
    \frac{Z_{N-M} / Z_N}{\Delta_M(Z) \Delta_M(W)} 
    \det_{1 \le \alpha, \beta \le M}
    \qty( \ee^{-V_L(z_\alpha)}\ee^{-V_R(w_\beta)}
    \sum_{i=N-M}^\infty \widetilde{\phi}_i(z_\alpha) \widetilde{\psi}_i(w_\beta)
    )    
    \nonumber \\
    & =    
    \frac{Z_{N-M} / Z_N}{\Delta_M(Z) \Delta_M(W)} \ee^{-\operatorname{tr}V_L(Z)}\ee^{-\operatorname{tr}V_R(W)}
    \det_{1 \le \alpha, \beta \le M}
    \qty( \widetilde{K}_{N-M}(w_\beta,z_\alpha) ) 
    \, .
\end{align}
This completes the derivation of the formula~\eqref{eq:ch_poly_pair_inv1}. \\
\noindent
We then consider the case $M \ge N$.
In this case, we similarly obtain the formula~\eqref{eq:ch_poly_pair_inv2} as follows,
\begin{align}
    & \expval{ \prod_{\alpha = 1}^M \det(z_\alpha - X_L)^{-1} \det(w_\alpha - X_R)^{-1} }
    \nonumber \\
    & =    
    \frac{Z_N^{-1}}{\Delta_M(Z) \Delta_M(W)}
    \nonumber \\ & \times
    \sum_{\substack{0 \le \lambda_N \le \cdots \le \lambda_1 \le \infty \\
    0 \le \mu_N \le \cdots \le \mu_1 \le \infty}}
    \det_{\substack{i=1,\ldots,N \\ \alpha = 1,\ldots,M \\ a = 1,\ldots,M-N}}
    \begin{pmatrix}
    z_\alpha^{-\lambda_i + i - (N+1)} \\ p_{a-1}(z_\alpha)
    \end{pmatrix}
    \det_{1 \le i,j \le N} ( x_L^{\lambda_i + N - i} \mid \omega \mid x_R^{\mu_j + N - j}) 
    \det_{\substack{j=1,\ldots,N \\ \beta = 1,\ldots,M \\ b = 1,\ldots,M-N}}
    \begin{pmatrix}
    w_\beta^{- \mu_j + j - (N+1)} \\ q_{b-1}(w_\beta)
    \end{pmatrix} \nonumber \\
    & = \frac{Z_N^{-1}}{\Delta_M(Z) \Delta_M(W)} \frac{1}{N!^2}
    \sum_{\substack{0 \leq r_1, \cdots, r_N \leq \infty \\ 0 \leq s_1, \cdots, s_N \leq \infty \\ r_i \neq r_j , s_i \neq s_j}}
    \det_{\substack{i=1,\ldots,N \\ \alpha = 1,\ldots,M \\ a = 1,\ldots,M-N}}
    \begin{pmatrix}
    z_\alpha^{- r_i -1} \\ p_{a-1}(z_\alpha)
    \end{pmatrix} 
    \det_{1 \le i,j \le N }
    (x_L^{r_i} \mid \omega \mid x_R^{s_j}) 
    \det_{\substack{j=1,\ldots,N \\ \beta = 1,\ldots,M \\ b = 1,\ldots,M-N}}
    \begin{pmatrix}
    w_\beta^{- s_j - 1} \\ q_{b-1}(w_\beta)
    \end{pmatrix} \nonumber \\
    & = \frac{Z_N^{-1}}{\Delta_M(Z) \Delta_M(W)} \det_{\substack{1 \leq \alpha,\beta \leq M \\ 1 \le a, b \le M-N }}
    \begin{pmatrix}
    \displaystyle \sum_{r,s = 0}^\infty z_{\alpha}^{-r-1}w_{\beta}^{-s-1} (x_L^{r} \mid \omega \mid x_R^{s}) & q_{b-1}(w_\beta) \\
    p_{a-1}(z_\alpha) & 0
    \end{pmatrix}
    \nonumber \\
    & = \frac{Z_N^{-1}}{\Delta_M(Z) \Delta_M(W)} \det_{\substack{1 \leq \alpha,\beta \leq M \\ 1 \le a, b \le M-N }}
    \begin{pmatrix}
    \qty( \frac{1}{z_\alpha - x_L} \mid \omega \mid \frac{1}{w_\beta - x_R}) & q_{b-1}(w_\beta) \\
    p_{a-1}(z_\alpha) & 0
    \end{pmatrix}
    \nonumber \\
    & = \frac{(-1)^{M-N} Z_N^{-1}}{\Delta_M(Z) \Delta_M(W)}
    \det_{1 \le \alpha, \beta \le M} \qty( \frac{1}{z_\alpha - x_L} \mid \omega \mid \frac{1}{w_\beta - x_R})
    \det_{1 \le a, b \le M - N}
    \qty( \sum_{\alpha,\beta = 1}^M p_{a-1}(z_\alpha) \widetilde{\omega}_{\alpha,\beta} q_{b-1}(w_\beta)  )
\end{align}
This completes the proof.
\end{proof}

\subsection{Mixed pair correlation}

We consider the mixed-type pair correlation function of the characteristic polynomials.
\begin{proposition}
Let $Z = \diag(z_1,\ldots,z_M)$ and $W = \diag(w_1,\ldots,w_M)$.
The following determinantal formulas hold for the mixed-pair correlation for $M \le N$.
\begin{subequations}
\begin{align}
    \expval{ \prod_{\alpha = 1}^M \det(z_\alpha - X_L) \det(w_\alpha - X_R)^{-1} }
    & = 
    \frac{Z_{N-M}/Z_N}{\Delta_M(Z) \Delta_M(W)}
    \det_{\substack{\alpha = 1, \ldots, M \\ \beta = 1,\ldots, 2M}}
    \begin{pmatrix}
    P_{N+M-\beta}(z_\alpha) \\
    \widetilde{Q}_{N+M-\beta}(w_\alpha) 
    \end{pmatrix}
    \, , \label{eq:ch_poly_av_mix1} \\
    \expval{ \prod_{\alpha = 1}^M \det(z_\alpha - X_L)^{-1} \det(w_\alpha - X_R) }
    & = 
    \frac{Z_{N-M}/Z_N}{\Delta_M(Z) \Delta_M(W)}
    \det_{\substack{\alpha = 1, \ldots, M \\ \beta = 1,\ldots, 2M}}
    \begin{pmatrix}
    \widetilde{P}_{N+M-\beta}(z_\alpha) \\
    Q_{N+M-\beta}(w_\alpha) 
    \end{pmatrix}
    \, . \label{eq:ch_poly_av_mix2} 
\end{align}
\end{subequations}
\end{proposition}
\begin{proof}
Applying the Schur polynomial expansion and the co-factor expansion as before, we obtain the following,
\begin{align}
    & \expval{ \prod_{\alpha = 1}^M \det(z_\alpha - X_L) \det(w_\alpha - X_R)^{-1} }
    \nonumber\\
    & = 
    \frac{Z_N^{-1}}{\Delta_M(Z) \Delta_M(W)}
    \sum_{\ell(\lambda) \le M} 
    \det_{1 \le \alpha, \beta \le M} \qty(w_\alpha^{-\lambda_\beta + \beta - N - 1}) 
    \det_{\substack{i=1,\ldots,N+M \\ \alpha,\beta = 1,\ldots,M \\ k = M+1,\ldots,N}}
    \begin{pmatrix}
    z_\alpha^{N+M-i} \\
    (x_L^{N+M-i} \mid \omega \mid x_R^{\lambda_\beta+N-\beta}) \\
    (x_L^{N+M-i} \mid \omega \mid x_R^{N-k}) 
    \end{pmatrix}
    \nonumber\\
    & = 
    \frac{Z_N^{-1}}{\Delta_M(Z) \Delta_M(W)}
    \sum_{\ell(\lambda) \le M} 
    \det_{1 \le \alpha, \beta \le M} \qty(w_\alpha^{-\lambda_\beta + \beta - N - 1}) 
    \det_{\substack{i=1,\ldots,N+M \\ \alpha,\beta = 1,\ldots,M \\ k = M+1,\ldots,N}}
    \begin{pmatrix}
    p_{N+M-i}(z_\alpha) \\
    (p_{N+M-i} \mid \omega \mid x_R^{\lambda_\beta+N-\beta}) \\
    (p_{N+M-i} \mid \omega \mid x_R^{N-k}) 
    \end{pmatrix}    
    \nonumber\\
    & = 
    \frac{Z_N^{-1}}{\Delta_M(Z) \Delta_M(W)}
    \det_{\substack{i=1,\ldots,N+M \\ \alpha,\beta = 1,\ldots,M \\ k = M+1,\ldots,N}}
    \begin{pmatrix}
    p_{N+M-i}(z_\alpha) \\
    \widetilde{q}_{N+M-i}(w_\beta) \\
    (p_{N+M-i} \mid \omega \mid x_R^{N-k}) 
    \end{pmatrix}        
    \, .
\end{align}
Then, the determinant part is given by
\begin{align}
    & \det_{\substack{i=1,\ldots,N+M \\ \alpha,\beta = 1,\ldots,M \\ k = M+1,\ldots,N}}
    \begin{pmatrix}
    p_{N+M-i}(z_\alpha) \\
    \widetilde{q}_{N+M-i}(w_\beta) \\
    (p_{N+M-i} \mid \omega \mid x_R^{N-k}) 
    \end{pmatrix}  
    \nonumber \\
    & = \det_{\substack{\alpha,\beta, \gamma, \delta = 1,\ldots,M \\ k,l = 1,\ldots,N-M}}
    \begin{pmatrix}
    p_{N+M-\gamma}(z_\alpha) & p_{N-\delta}(z_\alpha) & p_{N-M-l}(z_\alpha) \\
    \widetilde{q}_{N+M-\gamma}(w_\beta) & \widetilde{q}_{N-\delta}(w_\beta) & \widetilde{q}_{N-M-l}(w_\beta) \\
    (p_{N+M-\gamma} \mid \omega \mid q_{N-M-k}) & (p_{N-\delta} \mid \omega \mid q_{N-M-k}) & (p_{N-M-l} \mid \omega \mid q_{N-M-k}) 
    \end{pmatrix}
    \nonumber \\
    & = \det_{\substack{\alpha,\beta, \gamma, \delta = 1,\ldots,M \\ k,l = 1,\ldots,N-M}}
    \begin{pmatrix}
    P_{N+M-\gamma}(z_\alpha) & P_{N-\delta}(z_\alpha) & P_{N-M-l}(z_\alpha) \\
    \widetilde{Q}_{N+M-\gamma}(w_\beta) & \widetilde{Q}_{N-\delta}(w_\beta) & \widetilde{Q}_{N-M-l}(w_\beta) \\
    0 & 0 & h_{N-M-l} \, \delta_{N-M-l,N-M-k}
    \end{pmatrix}    
    \nonumber \\
    & = Z_{N-M}
    \det_{\alpha, \beta, \gamma, \delta = 1,\ldots,M}
    \begin{pmatrix}
    P_{N+M-\gamma}(z_\alpha) & P_{N-\delta}(z_\alpha)  \\
    \widetilde{Q}_{N+M-\gamma}(w_\beta) & \widetilde{Q}_{N-\delta}(w_\beta)
    \end{pmatrix}
    \, .
\end{align}
This completes the derivation of \eqref{eq:ch_poly_av_mix1}.
The other formula \eqref{eq:ch_poly_av_mix2} can be also derived in the same way.
\end{proof}

\begin{remark}
For $M = 1$, the mixed-pair correlation functions are given by
\begin{subequations}
\begin{align}
    \expval{ \frac{\det(z - X_L)}{\det(w - X_R)} } 
    & = \frac{Z_{N-1}}{Z_N} \det 
    \begin{pmatrix}
    P_{N} (z) & P_{N-1} (z) \\ \widetilde{Q}_N (w) & \widetilde{Q}_{N-1}(w)
    \end{pmatrix}
    \nonumber \\
    & = \frac{1}{h_{N-1}} 
    \left( P_N(z) \widetilde{Q}_{N-1}(w) - P_{N-1}(z) \widetilde{Q}_N(w) \right)
    \, , \\
    \expval{ \frac{\det(w - X_R)}{\det(z - X_L)} } 
    & = \frac{Z_{N-1}}{Z_N} \det 
    \begin{pmatrix}
    \widetilde{P}_{N} (z) & \widetilde{P}_{N-1} (z) \\ {Q}_N (w) & {Q}_{N-1}(w)
    \end{pmatrix}
    \nonumber \\
    & = \frac{1}{h_{N-1}} 
    \left( \widetilde{P}_N(z) {Q}_{N-1}(w) - \widetilde{P}_{N-1}(z) \widetilde{Q}_N(w) \right)
    \, .    
\end{align}
\end{subequations}
These expressions suggest that the mixed-pair correlation could be also written in terms of the associated CD kernel.
See~\cite{Strahov:2002zu,Baik:2003JMP,Borodin:2006CPAM,Eynard:2015aea} for details.
We leave this issue for the future study.
\end{remark}

\newpage

\bibliographystyle{amsalpha_mod}
\bibliography{ref}

\end{document}